\documentclass[journal,onecolumn,12pt]{IEEEtran}
\usepackage{pslatex}
\usepackage{amsfonts,color,morefloats}
\usepackage{amssymb,amsmath,latexsym,amsthm}

\newtheorem{theorem}{Theorem}
\newtheorem{lemma}[theorem]{Lemma}

\newtheorem{example}{Example}

\begin{document}

\title{The 2-Adic Complexity of Two Classes of Binary Sequences with Interleaved Structure \thanks{The work was supported by the National Natural Science Foundation of China (NSFC) under Grant 11701553.}}

\author{Shiyuan Qiang \thanks{Shiyuan Qiang is with the department of Applied Mathematics, China Agricultural, university, Beijing 100083, China (Email: qsycau\_18@163.com).}
Xiaoyan Jing  \thanks{Xiaoyan Jing is the Research Center for Number Theory and Its Applications, Northwest University, Xi'an 710127, China (Email: jxymg@126.com).}
Minghui Yang\thanks{Minghui Yang is with the State Key Laboratory of Information Security, Institute of Information Engineering, Chinese Academy of Sciences, Beijing 100093, China (Email: yangminghui6688@163.com).}
}

\date{\today}
\maketitle

\begin{abstract}
 The autocorrelation values of two classes of binary sequences are shown to be good in \cite{T}. We study the 2-adic complexity of these sequences. Our results show that the 2-adic complexity of such sequences is large enough to resist the attack of the rational approximation algorithm.
\end{abstract}


\section{Introduction}\label{sec-intro}

Pseudo-random sequences with low autocorrelation values, large linear complexity and so on have wide applications in cryptography and communication system. Due to the rational approximation algorithm \cite{K}, the 2-adic complexity has become an important pseudo-randomness index of binary sequences. Using sequences with low 2-adic complexity may pose a risk in cryptography and communication system. Therefore, it is meaningful to study the 2-adic complexity of binary sequences with large period, low autocorrelation values, etc. The 2-adic complexity of binary sequences has been determined in several papers, see \cite{H}, \cite{S1}, \cite{S2}, [7-9].

For a binary sequence $s=(s_0, s_1, \ldots, s_{N-1})$ with period $N$, the autocorrelation function is defined by
$$C_s(\tau)=\sum_{i=0}^{N-1}(-1)^{s_i+s_{i+\tau}}, \ \ \ \tau\in\mathbb{Z}/N\mathbb{Z}.$$

Interleaved operator introduced by Gong \cite{G} is a powerful tool to construct sequences with low autocorrelation and large period. Let ${s^0, s^1,\ldots, s^{T-1}}$ be $T$ sequences with period $N$, where $s^i=(s^{(i)}_0, s^{(i)}_1,\ldots, s^{(i)}_{N-1})\\(0\leq i< T).$ Construct the matrix
\begin{equation*}U=
\left(
\begin{array}{cccc}
s_0^{(0)} & s_0^{(1)} & \cdots & s_0^{(T-1)}\\
s_1^{(0)} & s_1^{(1)} & \cdots & s_1^{(T-1)}\\
\vdots & \vdots &\ddots &  \vdots\\
s_{N-1}^{(0)} & s_{N-1}^{(1)}& \cdots & s_{N-1}^{(T-1)}
\end{array}
\right)
\end{equation*}
by placing the sequence $s^i$ on the $i$th column and concatenate the successive rows of the matrix $U$. The interleaved sequence $u=(u_{i T+j})=(U_{i, j}) (0 \leq i <N, 0 \leq j <T$) with period $NT$ is denoted by $u=I(s^0, s^1,\ldots, s^{T-1})$ for simplicity.

\section{Preliminaries}\label{sec2}

The following sequences are shown to have low autocorrelation values.

\begin{lemma}\cite{T}\label{lem1}
Let $s_A=I(\mathbf{1}_{2^k-1},a_1,\cdots,a_{2^k})$ be the modified generalized GMW sequences of period $N_A=2^{2k}-1$, where $\mathbf{1}_{2^k-1}$ is the all one sequence of period $2^k-1$ and $a_i(1\leq i\leq 2^k)$ is some shift of a 2-level autocorrelation sequence $a$ of period $2^k-1$. Then the autocorrelation value of $s_A$ is given by
\begin{align*}
C_{s_A}=\left\{ \begin{array}{ll}
2^{2k}-1, & \textrm{if $\tau=0$}\\
-1,& \textrm{if $\tau\equiv0\pmod{2^k+1}$ and $\tau\neq0$}\\
3,& \textrm{otherwise}.
\end{array} \right.
\end{align*}
\end{lemma}

\begin{lemma}\cite{T}\label{lem2}
Let $p$ and $p+2$ be two primes. The modified two-prime sequence is defined by
\begin{align*}
s_{B}(i)=\left\{ \begin{array}{ll}
1, & \textrm{if $i\equiv0\pmod{p+2}$}\\
1,& \textrm{if $i\equiv0\pmod{p}$ and $i\neq0$}\\
\frac{1-(\frac{i}{p})(\frac{i}{p+2})}{2},& \textrm{otherwise}
\end{array} \right.
\end{align*}
where $(\frac{\cdot}{\cdot})$ denotes the Legendre symbol. Then the autocorrelation value of $s_B$ is given by
\begin{align*}
C_{s_B}=\left\{ \begin{array}{ll}
p(p+2), & \textrm{if $\tau=0$}\\
-1,& \textrm{if $\tau\equiv0\pmod{p+2}$ and $\tau\neq0$}\\
3,& \textrm{otherwise}
\end{array} \right.
\end{align*}
\end{lemma}

Let $s=(s_0, s_1, \ldots, s_{N-1})$ be a binary sequence with period $N$. Denote $S(x)=s_0+s_1x+\cdots+s_{N-1}x^{N-1}$.
Then the 2-adic complexity $\Phi_2(s)$ \cite{K}
of the binary sequence $s$ is defined by $\log_2\frac{2^N-1}{\gcd(2^N-1, S(2))}$, where $\gcd(a,b)$ denotes the greatest common divisor of $a$ and $b$.

\section{Main result}

In this section we will investigate the 2-adic complexity of $s_A$ and $s_B$. Firstly, we determine the lower bound of the 2-adic complexity of $s_A$. Then we determine the exact value of the 2-adic
complexity of $s_B$.

The following lemma is useful in the sequel.

\begin{lemma}\label{lem3}
Let $s=(s_0,s_1,\cdots,s_{N-1})$ be a binary sequence of period N, $S(x)=\sum_{i=0}^{N-1}s_ix^i\in  \mathbb{Z} [x] $ and $T(x)=\sum_{i=0}^{N-1}(-1)^{s_i}x^i\in \mathbb{Z}[x]$. Then
\begin{align*}
-2S(2)T(2^{-1})\equiv N+\sum_{\tau=1}^{N-1}C_s(\tau)2^{\tau}\pmod{2^N-1}
\end{align*}
\end{lemma}

\begin{theorem}\label{th5} Let the symbols be the same as before. Then
the 2-adic complexity $\Phi_2(S_A)$ of $S_A$ with period $N_A=2^{2k}-1$ satisfies $\Phi_2(s_A)>(2^{2k}-1)-1-2(k-1)> \frac{N_A}{2} $.
\end{theorem}
\begin{proof}
From Lemmas \ref{lem1} and \ref{lem3} we have
\begin{align}
-2S_A(2)T_A(2^{-1})&\equiv N_A+\sum_{\tau=1}^{N_A-1}C_{s_A}(\tau)2^{\tau} \notag\\
&\equiv2^{2k}-1+\sum_{\tau=1}^{N_A-1}3\cdot2^{\tau}-\sum_{\tau=1}^{2^k-1-1}4\cdot2^{(2^k+1)\tau}\notag\\
                                              &\equiv2^{2k}-1+3\bigg(\frac{1-2^{N_A}}{1-2}-1\bigg)-4\bigg(\frac{1-2^{(2^k+1)(2^k-1)}}{1-2^{2^k+1}}-1\bigg)\notag\\
                                              &\equiv2^{2k}-4\frac{2^{2^{2k}-1}-1}{2^{2^k+1}-1}\pmod{2^{2^{2k}-1}-1}.\label{equ1}
\end{align}

Let $\gcd(S_A(2)T_A(2^{-1}), 2^{2^{2k}-1}-1)=d_A.$ Then by (\ref{equ1}) we get
\begin{align*}
d_A&=\gcd(2^{2k}-4\frac{2^{2^{2k}-1}-1}{2^{2^k+1}-1},2^{2^{2k}-1}-1)\\
&\leq\gcd(2^{2k}-4\frac{2^{2^{2k}-1}-1}{2^{2^k+1}-1},2^{2^{k}+1}-1)\cdot \gcd(2^{2k}-4\frac{2^{2^{2k}-1}-1}{2^{2^k+1}-1},\frac{2^{2^{2k}-1}-1}{2^{2^{k}+1}-1})\\
 &=\gcd(2^{2k}-4\frac{2^{2^{2k}-1}-1}{2^{2^k+1}-1},2^{2^{k}+1}-1)\\
 &=\gcd(2^{2k}-4\sum_{i=0}^{2^k-2}2^{(2^k+1)i},2^{2^{k}+1}-1)
\end{align*}

And from
\begin{align*}
2^{2k}-4\sum_{i=0}^{2^k-2}2^{(2^k+1)i}&\equiv2^{2k}-4(2^k-2+1)\pmod{2^{2^k+1}-1}\\
                                      &\equiv(2^k-2)^2\pmod{2^{2^k+1}-1},
\end{align*}
we know
$$d_A\leq\gcd((2^k-2)^2,2^{2^k+1}-1)=\gcd((2^{k-1}-1)^2,2^{2^k+1}-1).$$
It then follows that $\gcd(S_A(2),2^{N_A}-1)\leq d_A \leq(2^{k-1}-1)^2$.

Hence from the definition of $\Phi_2(s_A)$ we get
\begin{align*}
\Phi_2(s_A)&=\log_2\frac{2^{2^{2k}}-1}{\gcd(S_A(2),2^{N_A}-1)}\geq\log_2\frac{2^{N_A}-1}{(2^{k-1}-1)^2}\\
           &=\log_2(2^{2^{2k}-1}-1)-2\log_2(2^{k-1}-1)\\
           &>(2^{2k}-1)-1-2(k-1).
\end{align*}

Let $f(k)=(2^{2k}-1)-2[2(k-1)+1]=4^{k}-4k+1$. It is obvious that $f(k)> 0$. Therefore we get $(2^{2k}-1)-1-2(k-1)> \frac{2^{2k}-1}{2}$ which implies that $\Phi_2(s_A)> \frac{N_A}{2}$.
\end{proof}

In the following we will denote $c=p(p+2)-4\frac{2^{p(p+2)}-1}{2^{p+2}-1}+1$ and assume that $p$ and  $p+2$ are both prime without loss of generality.

\begin{lemma}\label{lem4}
$\gcd(c, 2^p-1)=1$.
\end{lemma}
\begin{proof}
From $2^p-1\mid2^{p(p+2)}-1$ and $\gcd(2^p-1, 2^{p+2}-1)=1$, we get $2^p-1\mid\frac{2^{p(p+2)}-1}{2^{p+2}-1}$. It then follows that $$\gcd(c, 2^p-1)=\gcd(p(p+2)+1,2^p-1)=\gcd((p+1)^2,2^p-1).$$

Assume that $q$ is a prime divisor of $\gcd(c, 2^p-1)$, then $2^p\equiv1\pmod q$ and $2^{q-1}\equiv1\pmod q$.
Therefore $q\mid p+1, p\mid q-1$.
Let $q-1=pk$, then $q=pk+1\mid p+1$, we have $k=1, q=p+1$ which is a contradiction.
\end{proof}

\begin{lemma}\label{lem5}
$\gcd(c,2^{p+2}-1)=1$.
\end{lemma}
\begin{proof}
From
\begin{align*}
c=&p(p+2)-4\frac{2^{p(p+2)}-1}{2^{p+2}-1}+1\\
\equiv&p(p+2)-4\sum_{i=0}^{p-1}2^{(p+2)i}+1\\
\equiv&p(p+2)-4p+1\pmod{2^{p+2}-1}
\end{align*}
we get $$\gcd(c,2^{p+2}-1)=\gcd(p(p+2)-4p+1,2^{p+2}-1)=\gcd((p-1)^2,2^{p+2}-1).$$

Assume that $q$ is a prime divisor of $\gcd(c,2^{p+2}-1)$, then $2^{p+2}\equiv1\pmod q$ and $2^{q-1}\equiv1\pmod q$. Therefore $q\mid p-1, p+2\mid q-1$.
Let $p-1=kq$, then $p+2=kq+3\mid q-1$ which is a contradiction. This implies that $\gcd(c,2^{p+2}-1)=1.$
\end{proof}

\begin{lemma}\label{lem6}
$\gcd(c,\frac{2^{p(p+2)}-1}{(2^p-1)(2^{p+2}-1)})=1$.
\end{lemma}
\begin{proof}
From $\frac{2^{p(p+2)}-1}{(2^p-1)(2^{p+2}-1)}\mid\frac{2^{p(p+2)}-1}{2^{p+2}-1}$, we have $\gcd(c,\frac{2^{p(p+2)}-1}{(2^p-1)(2^{p+2}-1)})=\gcd(p(p+2)+1,\frac{2^{p(p+2)}-1}{(2^p-1)(2^{p+2}-1)})=\gcd((p+1)^2,\frac{2^{p(p+2)}-1}{(2^p-1)(2^{p+2}-1)})$.
Assume that $q$ is a prime divisor of $\gcd(c,\frac{2^{p(p+2)}-1}{(2^p-1)(2^{p+2}-1)})$, then $2^{p(p+2)}\equiv1\pmod q$ and $2^{q-1}\equiv1\pmod q$. Therefore the order of 2 mod $q$ is $p$, $p+2$ or $p(p+2)$.

If the order of 2 mod $q$ is $p$, then the result follows from Lemma 5.

If the order of 2 mod $q$ is $p+2$, then the result follows from Lemma 6.

If the order of 2 mod $q$ is $p(p+2)$, $p(p+2)\mid q-1$ which contradicts to $q\mid p+1$.
\end{proof}

\begin{theorem}\label{th9}
Let the symbols be the same as before. Then the 2-adic complexity of $s_B$ with period $N_B=p(p+2)$ is $$\Phi_2(s_B)=N_B=\log_2(2^{p(p+2)}-1).$$
\end{theorem}
\begin{proof}
From Lemmas \ref{lem2} and \ref{lem3} we have
\begin{align}
-2S_B(2)T_B(2^{-1})&\equiv N_B+\sum_{\tau=1}^{N_B-1}C_{s_B}(\tau)2^{\tau}\notag\\
&=p(p+2)+\sum_{\tau=1}^{N_B-1}3\cdot2^{\tau}-4\sum_{i=1}^{p-1}2^{(p+2)i}\notag\\
                                              &=p(p+2)+3\big(\frac{1-2^{p(p+2)}}{1-2}-1\big)-4\big(\frac{1-2^{p(p+2)}}{1-2^{p+2}}-1\big)\notag\\
                                              &\equiv p(p+2)-4\frac{2^{p(p+2)}-1}{2^{p+2}-1}+1\pmod{2^{N_B}-1}.\label{equ2}
\end{align}
Hence by (2) we get $$\gcd(S_B(2), 2^{N_B}-1)\leq \gcd(c, 2^p-1)\gcd(c, 2^{p+2}-1)\gcd(c, \frac{2^{p(p+2)}-1}{(2^p-1)(2^{p+2}-1)}).$$
Then the result follows from Lemmas 5, 6, 7.
\end{proof}

Finally, we give an example to illustrate Theorem 8.
\begin{example}
For two primes  $p=3$ and $p+2=5$, according to the definition of the modified two-prime sequence, we get $$s_B=(1 0 0 1 0 1 1 1 0 1 1 1 1 1 1)$$with period $N_B=p(p+2)=15$.

Then $s_B(2)=1\cdot2^0+0\cdot2^1+\cdots+1\cdot2^{14}=32489=53\cdot613$.

For $2^{N_B}-1=2^{15}-1=32767=7\cdot31\cdot151, \Phi_2(s_B)=\log_2\frac{2^{N_B}-1}{\gcd(2^{N_B}-1, S(2))}=\log_2\frac{7\cdot31\cdot151}{\gcd(7\cdot31\cdot151,53\cdot613)}=\log_2(2^{15}-1).$ The results are consistent
with the Theorem 8.
\end{example}

\end{document}